\documentclass{article}%
\usepackage{amssymb}
\usepackage{amsmath}
\usepackage{amsfonts}
\usepackage{graphicx}%
\providecommand{\U}[1]{\protect\rule{.1in}{.1in}}
\newtheorem{theorem}{Theorem}[section]

\newtheorem{corollary}[theorem]{Corollary}

\newtheorem{lemma}[theorem]{Lemma}

\newtheorem{remark}[theorem]{Remark}

\newenvironment{proof}[1][Proof]{\noindent\textbf{#1.} }{\ \rule{0.5em}{0.5em}}
\begin{document}

\author{Vadim E. Levit\\Ariel University Center of Samaria, Israel\\levitv@ariel.ac.il
\and Eugen Mandrescu\\Holon Institute of Technology, Israel\\eugen\_m@hit.ac.il}
\title{On the Core of a Unicyclic Graph}
\date{}
\maketitle

\begin{abstract}
A set $S\subseteq V$ is \textit{independent} in a graph $G=\left(  V,E\right)
$ if no two vertices from $S$ are adjacent. By \textrm{core}$\left(  G\right)
$ we mean the intersection of all maximum independent sets. The
\textit{independence number} $\alpha(G)$ is the cardinality of a maximum
independent set, while $\mu(G)$ is the size of a maximum matching in $G$.

A connected graph having only one cycle, say $C$, is a \textit{unicyclic
graph}. In this paper we prove that if $G$ is a unicyclic graph of order $n$
and $n-1=\alpha(G)+\mu(G)$, then $\mathrm{core}\left(  G\right)  $ coincides
with the union of cores of all trees in $G-C$.


\end{abstract}

\section{Introduction}

Throughout this paper $G=(V,E)$ is a simple (i.e., a finite, undirected,
loopless and without multiple edges) graph with vertex set $V=V(G)$ and edge
set $E=E(G)$. If $X\subset V$, then $G[X]$ is the subgraph of $G$ spanned by
$X$. By $G-W$ we mean the subgraph $G[V-W]$, if $W\subset V(G)$. For $F\subset
E(G)$, by $G-F$ we denote the partial subgraph of $G$ obtained by deleting the
edges of $F$, and we use $G-e$, if $W$ $=\{e\}$. If $A,B$ $\subset V$ and
$A\cap B=\emptyset$, then $(A,B)$ stands for the set $\{e=ab:a\in A,b\in
B,e\in E\}$. The neighborhood of a vertex $v\in V$ is the set $N(v)=\{w:w\in
V$ \textit{and} $vw\in E\}$, and $N(A)=\cup\{N(v):v\in A\}$, $N[A]=A\cup N(A)$
for $A\subset V$. By $C_{n},K_{n}$ we mean the chordless cycle on $n\geq$ $4$
vertices, and respectively the complete graph on $n\geq1$ vertices.

A set $S$ of vertices is \textit{independent} if no two vertices from $S$ are
adjacent, and an independent set of maximum size will be referred to as a
\textit{maximum independent set}. The \textit{independence number }of $G$,
denoted by $\alpha(G)$, is the size of a maximum independent set\textit{\ }of
$G$. Let $\Omega(G)$ denote the family $\{S:S$ \textit{is a maximum
independent set of} $G\}$, while
\[
\mathrm{core}(G)=\cap\{S:S\in\Omega(G)\}\text{ \cite{levm3}.}%
\]
An edge $e\in E(G)$ is $\alpha$-\textit{critical} whenever $\alpha
(G-e)>\alpha(G)$. Notice that the inequalities $\alpha(G)\leq\alpha
(G-e)\leq\alpha(G)+1$ hold for each edge $e$.

A matching (i.e., a set of non-incident edges of $G$) of maximum cardinality
$\mu(G)$ is a \textit{maximum matching}, and a \textit{perfect matching} is
one covering all vertices of $G$. An edge $e\in E(G)$ is $\mu$%
-\textit{critical }provided $\mu(G-e)<\mu(G)$.

\begin{theorem}
\label{th2}\cite{LevMan3} For every graph $G$ no $\alpha$-critical edge has an
endpoint in $N[\mathrm{core}(G)]$.
\end{theorem}

It is well-known that
\[
\lfloor n/2\rfloor+1\leq\alpha(G)+\mu(G)\leq n
\]
hold for every graph $G$ with $n$ vertices. If $\alpha(G)+\mu(G)=n$, then $G$
is called a \textit{K\"{o}nig-Egerv\'{a}ry graph }\cite{dem}, \cite{ster}.
Several properties of K\"{o}nig-Egerv\'{a}ry graphs are presented in
\cite{Gavril}, \cite{KoNgPeis}, \cite{Larson2011}, \cite{levm4},
\cite{LevMan09}, \cite{LevManKE2009}.

It is known that every bipartite graph is a K\"{o}nig-Egerv\'{a}ry\emph{
}graph as well \cite{eger}, \cite{koen}. This class includes also
non-bipartite graphs (see, for instance, the graph $G$ in Figure
\ref{fig112}).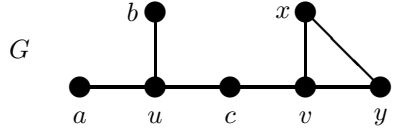
\begin{figure}[h]
\setlength{\unitlength}{1cm}\begin{picture}(5,1.8)\thicklines
\multiput(4,0.5)(1,0){5}{\circle*{0.29}}
\multiput(5,1.5)(2,0){2}{\circle*{0.29}}
\put(4,0.5){\line(1,0){4}}
\put(5,0.5){\line(0,1){1}}
\put(7,1.5){\line(1,-1){1}}
\put(7,0.5){\line(0,1){1}}
\put(4,0.1){\makebox(0,0){$a$}}
\put(4.7,1.5){\makebox(0,0){$b$}}
\put(6,0.1){\makebox(0,0){$c$}}
\put(5,0.1){\makebox(0,0){$u$}}
\put(7,0.1){\makebox(0,0){$v$}}
\put(6.7,1.5){\makebox(0,0){$x$}}
\put(8,0.1){\makebox(0,0){$y$}}
\put(3.2,1){\makebox(0,0){$G$}}
\end{picture}
\caption{A K\"{o}nig-Egerv\'{a}ry graph with $\alpha(G)=\left\vert \left\{
a,b,c,x\right\}  \right\vert $ and $\mu(G)=\left\vert \left\{
au,cv,xy\right\}  \right\vert $.}%
\label{fig112}%
\end{figure}

\begin{theorem}
\label{th1} If $G$ is a K\"{o}nig-Egerv\'{a}ry graph, then

\emph{(i)} \cite{levm4} every maximum matching matches $N($\textrm{core}$(G))
$ into \textrm{core}$(G)$;

\emph{(ii) }\cite{LevMan3} $H=G-N[$\textrm{core}$(G)]$ is a
K\"{o}nig-Egerv\'{a}ry\emph{ }graph with a perfect matching and each maximum
matching of $H$ can be enlarged to a maximum matching of $G$.
\end{theorem}

The graph $G$ is called \textit{unicyclic} if it is connected and has a unique
cycle, which we denote by $C=\left(  V(C),E\left(  C\right)  \right)  $. Let%
\[
N_{1}(C)=\{v:v\in V\left(  G\right)  -V(C),N(v)\cap V(C)\neq\emptyset\},
\]
and $T_{x}=(V_{x},E_{x})$ be the tree of $G-xy$ containing $x$, where $x\in
N_{1}(C),y\in V(C)$.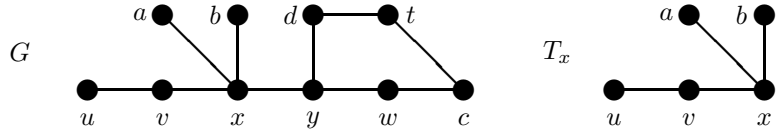
\begin{figure}[h]
\setlength{\unitlength}{1cm}\begin{picture}(5,1.8)\thicklines
\multiput(2,0.5)(1,0){6}{\circle*{0.29}}
\multiput(3,1.5)(1,0){4}{\circle*{0.29}}
\put(2,0.5){\line(1,0){5}}
\put(3,1.5){\line(1,-1){1}}
\put(4,0.5){\line(0,1){1}}
\put(5,1.5){\line(1,0){1}}
\put(6,1.5){\line(1,-1){1}}
\put(5,0.5){\line(0,1){1}}
\put(2,0.1){\makebox(0,0){$u$}}
\put(3,0.1){\makebox(0,0){$v$}}
\put(4,0.1){\makebox(0,0){$x$}}
\put(5,0.1){\makebox(0,0){$y$}}
\put(6,0.1){\makebox(0,0){$w$}}
\put(7,0.1){\makebox(0,0){$c$}}
\put(2.7,1.5){\makebox(0,0){$a$}}
\put(3.7,1.5){\makebox(0,0){$b$}}
\put(4.7,1.5){\makebox(0,0){$d$}}
\put(6.3,1.5){\makebox(0,0){$t$}}
\put(1.1,1){\makebox(0,0){$G$}}
\multiput(9,0.5)(1,0){3}{\circle*{0.29}}
\multiput(10,1.5)(1,0){2}{\circle*{0.29}}
\put(9,0.5){\line(1,0){2}}
\put(10,1.5){\line(1,-1){1}}
\put(11,0.5){\line(0,1){1}}
\put(9,0.1){\makebox(0,0){$u$}}
\put(10,0.1){\makebox(0,0){$v$}}
\put(11,0.1){\makebox(0,0){$x$}}
\put(9.7,1.5){\makebox(0,0){$a$}}
\put(10.7,1.5){\makebox(0,0){$b$}}
\put(8.25,1){\makebox(0,0){$T_{x}$}}
\end{picture}\caption{$G$ is a unicyclic non-K\"{o}nig-Egerv\'{a}ry graph with
$V(C)=\{y,d,t,c,w\}$.}%
\label{fig5333}%
\end{figure}

Unicyclic graphs keep enjoying plenty of interest, as one can see, for
instance, in \cite{Belardo2010}, \cite{Du2010}, \cite{Huo2010},
\cite{LevMan2009a}, \cite{Li2010}, \cite{Wu2010}, \cite{Zhai2010}.

In this paper we analyze the structure of $\mathrm{core}(G)$ for a unicyclic
graph $G$.

\section{Results}

If $G$ is a unicyclic graph, then there is an edge $e\in E\left(  C\right)  $,
such that $\mu(G-e)=\mu(G)$, because for each pair of edges, consecutive on
$C$, at most one could be $\mu$-critical. Let us mention that $\alpha
(G)\leq\alpha(G-e)\leq\alpha(G)+1$ holds for each edge $e\in E\left(
G\right)  $. Every edge of the unique cycle could be $\alpha$-critical; e.g.,
the graph $G$ from Figure \ref{fig5333}, which has also additional $\alpha
$-critical edges (e.g., the edge $uv$).

Let us notice that the bipartite graph $T_{x}$ from Figure \ref{fig5333} has
only two maximum matchings, namely, $M_{1}=\left\{  ax,uv\right\}  $ and
$M_{1}=\left\{  bx,uv\right\}  $, while each vertex of $\mathrm{core}%
(T_{x})=\{a,b\}$ is not saturated by one of these matchings.

\begin{lemma}
\label{lem0}For every bipartite graph $G$, a vertex $v\in\mathrm{core}(G)$ if
and only if there exists a maximum matching that does not saturate $v$.
\end{lemma}

\begin{proof}
Since $v\in\mathrm{core}(G)$, it follows that $\alpha(G-v)=\alpha(G)-1$.
Consequently, we have
\[
\alpha(G)+\mu(G)-1=\left\vert V(G)\right\vert -1=\left\vert V(G-v)\right\vert
=\alpha(G-v)+\mu(G-v)
\]
which implies that $\mu(G)=\mu(G-v)$. In other words, there is a maximum
matching in $G$ not saturating $v$.

Conversely, suppose that there exists a maximum matching in $G$ that does not
saturate $v$. Since, by Theorem \ref{th1}\emph{(i)}, $N($\textrm{core}$(G))$
is matched into \textrm{core}$(G)$ by every maximum matching, it follows that
$v\notin N($\textrm{core}$(G))$.

Assume that $v\notin$ \textrm{core}$(G)$. By Theorem \ref{th1}\emph{(ii)},
every maximum matching $M$ of $G$ is of the form $M=M_{1}\cup M_{2}$, where
$M_{1}$ matches $N($\textrm{core}$(G))$ into \textrm{core}$(G)$, while $M_{2}$
is a perfect matching of $G-N\left[  \mathrm{core}(G)\right]  $. Thus $v$ is
saturated by every maximum matching of $G$, in contradiction with the
hypothesis on $v$.
\end{proof}

\begin{remark}
Lemma \ref{lem0} fails for non-bipartite K\"{o}nig-Egerv\'{a}ry graphs; e.g.,
every maximum matching of the graph $G$ from Figure \ref{fig112} saturates
$c\in$ \textrm{core}$(G)=\{a,b,c\}$.
\end{remark}

\begin{lemma}
\label{lem2}If $G$ is a unicyclic graph of order $n$, then $n-1\leq
\alpha(G)+\mu(G)\leq n$.
\end{lemma}

\begin{proof}
If $e=xy\in E(C)$, then $G-e$ is a tree, because $G$ is connected. Hence,
$\alpha(G-e)+\mu(G-e)=n$. Clearly, $\alpha(G-e)\leq\alpha(G)+1$, while
$\mu(G-e)\leq\mu(G)$. Consequently, we get that
\[
n=\alpha(G-e)+\mu(G-e)\leq\alpha(G)+\mu(G)+1,
\]
which leads to $n-1\leq\alpha(G)+\mu(G)$. The inequality $\alpha(G)+\mu(G)\leq
n$ is true for every graph $G$.
\end{proof}

\begin{remark}
If $G$ has $n$ vertices, $p$ connected components, say $H_{i},1\leq i\leq p$,
and each component contains only one cycle, then one can easily see that
$n-p\leq\alpha(G)+\mu(G)\leq n$, because $\alpha(G)=\sum\limits_{i=1}%
^{p}\alpha(H_{i})$ and $\mu(G)=\sum\limits_{i=1}^{p}\mu(H_{i})$.
\end{remark}

While $C_{2k}$, $k\geq2$, has no $\alpha$-critical edge at all, each edge of
every odd cycle $C_{2k-1}$, $k\geq2$, is $\alpha$-critical. This property is
partially inherited by unicyclic graphs.

\begin{lemma}
\label{lem6}Let $G$ be a unicyclic graph of order $n$. Then $n-1=\alpha
(G)+\mu(G)$ if and only if each edge of its unique cycle is $\alpha$-critical.
\end{lemma}

\begin{proof}
Assume that $n-1=\alpha(G)+\mu(G)$. Since $G$\ is connected, for each $e\in
E(C)$ the graph $G-e$ is a tree. Hence, we have
\[
\alpha(G-e)-\alpha(G)+\mu(G-e)-\mu(G)=1,
\]
which implies $\mu(G-e)=\mu(G)$ and $\alpha(G-e)=\alpha(G)+1$, since%
\[
-1\leq\mu(G-e)-\mu(G)\leq0\leq\alpha(G-e)-\alpha(G)\leq1.
\]
In other words, every $e\in E(C)$ is $\alpha$-critical.

Conversely, let $e\in E\left(  C\right)  $ be such that $\mu(G-e)=\mu(G)$;
such an edge exists, because no two consecutive edges on $C$ could be $\mu
$-critical. Since $e$ is $\alpha$-critical, and $G-e$ is a tree, we infer
that
\[
n-1=\alpha(G-e)+\mu(G-e)-1=\alpha(G)+\mu(G),
\]
and this completes the proof.
\end{proof}

Combining Lemma \ref{lem6} and Theorem \ref{th2}, we infer the following.

\begin{corollary}
\label{cor1}If $G$ is a unicyclic non-K\"{o}nig-Egerv\'{a}ry graph, then no
vertex of its unique cycle belongs to $N[\mathrm{core}(G)]$.
\end{corollary}

\begin{remark}
Corollary \ref{cor1} is true also for some unicyclic K\"{o}nig-Egerv\'{a}ry
graphs; e.g., the graph $H_{1}$ from Figure \ref{fig1123}. However, the
K\"{o}nig-Egerv\'{a}ry graph $H_{2}$ from the same figure satisfies
$N[\mathrm{core}(H_{2})]\cap V\left(  C\right)  =\left\{  u\right\}
\neq\emptyset$.
\end{remark}

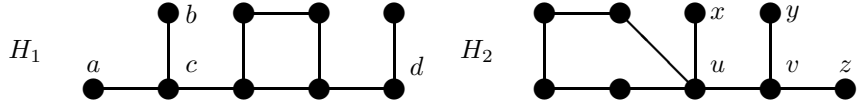
\begin{figure}[h]
\setlength{\unitlength}{1cm}\begin{picture}(5,1.3)\thicklines
\multiput(2,0)(1,0){5}{\circle*{0.29}}
\multiput(3,1)(1,0){4}{\circle*{0.29}}
\put(2,0){\line(1,0){4}}
\put(3,0){\line(0,1){1}}
\put(4,0){\line(0,1){1}}
\put(4,1){\line(1,0){1}}
\put(5,0){\line(0,1){1}}
\put(6,0){\line(0,1){1}}
\put(2,0.3){\makebox(0,0){$a$}}
\put(3.3,1){\makebox(0,0){$b$}}
\put(3.3,0.3){\makebox(0,0){$c$}}
\put(6.3,0.3){\makebox(0,0){$d$}}
\put(1.1,0.5){\makebox(0,0){$H_{1}$}}
\multiput(8,0)(1,0){5}{\circle*{0.29}}
\multiput(8,1)(1,0){4}{\circle*{0.29}}
\put(8,0){\line(1,0){4}}
\put(8,0){\line(0,1){1}}
\put(8,1){\line(1,0){1}}
\put(10,0){\line(-1,1){1}}
\put(10,0){\line(0,1){1}}
\put(11,0){\line(0,1){1}}
\put(10.3,1){\makebox(0,0){$x$}}
\put(11.3,1){\makebox(0,0){$y$}}
\put(12,0.3){\makebox(0,0){$z$}}
\put(10.3,0.3){\makebox(0,0){$u$}}
\put(11.3,0.3){\makebox(0,0){$v$}}
\put(7.1,0.5){\makebox(0,0){$H_{2}$}}
\end{picture}\caption{$H_{1}$ and $H_{2}$ have $N\left[  \mathrm{core}%
(H_{1})\right]  =\left\{  a,b,c\right\}  $, $N\left[  \mathrm{core}%
(H_{1})\right]  =\left\{  x,y,z,u,v\right\}  $.}%
\label{fig1123}%
\end{figure}

\begin{lemma}
\label{lem5}Let $G$ be a unicyclic graph of order $n$. If there exists some
$x\in N_{1}(C)$, such that $x\in\mathrm{core}(T_{x})$, then $G$ is a
K\"{o}nig-Egerv\'{a}ry graph.
\end{lemma}

\begin{proof}
Let $x\in\mathrm{core}(T_{x})$, $y\in N\left(  x\right)  \cap V(C)$, and $z\in
N\left(  y\right)  \cap V(C)$. Suppose, to the contrary, that $G$ is not a
K\"{o}nig-Egerv\'{a}ry graph. By\emph{ }Lemmas \ref{lem2} and \ref{lem6}, the
edge $yz$ is $\alpha$-critical. Hence $y\notin\mathrm{core}(G)$, which implies
that $\alpha(G)=\alpha(G-y)$. In accordance with Lemma \ref{lem0}, there
exists a maximum matching $M_{x}$ of $T_{x}$ not saturating $x$. Combining
$M_{x}$ with a maximum matching of $G-y-T_{x}$ we get a maximum matching
$M_{y}$ of $G-y$. Hence $M_{y}\cup\left\{  xy\right\}  $ is a matching of $G$,
which results in $\mu\left(  G\right)  \geq\mu\left(  G-y\right)  +1$.
Therefore, using Lemma \ref{lem2} and having in mind that $G-y$ is a forest of
order $n-1$, we get the following contradiction
\[
n-1=\alpha(G)+\mu\left(  G\right)  \geq\alpha(G-y)+\mu\left(  G-y\right)
+1=n-1+1=n,
\]
that completes the proof.
\end{proof}

\begin{remark}
The converse of Lemma \ref{lem5} is not generally true; e.g., the graph
$H_{1}$ from Figure \ref{fig1123} is a unicyclic K\"{o}nig-Egerv\'{a}ry graph,
while both $c\notin\mathrm{core}(T_{c})=\left\{  a,b\right\}  $, and
$d\notin\mathrm{core}(T_{d})=\emptyset$.
\end{remark}

\begin{theorem}
\label{th22}If $G$ is a unicyclic non-K\"{o}nig-Egerv\'{a}ry graph, then
\[
\mathrm{core}\left(  G\right)  =\cup\left\{  \mathrm{core}\left(
T_{x}\right)  :x\in N_{1}(C)\right\}  .
\]

\end{theorem}

\begin{proof}
\textit{Claim 1}. Every maximum independent set of $T_{x}$ may be enlarged to
some maximum independent set of $G$, for each $x\in N_{1}(C)$.

Let $A\in\Omega(T_{x})$, $y\in N\left(  x\right)  \cap V(C)$, and $z\in
N\left(  y\right)  \cap V(C)$. According to Lemma \ref{lem6}, the edge $yz$ is
$\alpha$-critical. Hence there exist $S_{y}\in\Omega(G)$, $S_{yz}\in$
$\Omega(G-yz)$, such that $y\in S_{y}$ and $y,z\in S_{yz}$.

\textit{Case 1}. Assume that $x\notin A$.

If $\left\vert S_{y}-V(T_{x})\right\vert <\alpha(G-T_{x})=\left\vert
S_{0}\right\vert $, where $S_{0}\in\Omega\left(  G-T_{x}\right)  $, then the
set $S_{1}=S_{0}\cup\left(  S_{y}\cap V(T_{x})\right)  $ is independent in
$G$, and we get the contradiction
\[
\alpha\left(  G\right)  =\left\vert S_{y}-V(T_{x})\right\vert +\left\vert
S_{y}\cap V(T_{x})\right\vert <\left\vert S_{0}\right\vert +\left\vert
S_{y}\cap V(T_{x})\right\vert =\left\vert S_{1}\right\vert .
\]
Therefore, we have $\left\vert S_{y}-V(T_{x})\right\vert =\alpha(G-T_{x})$.
Then $A\cup\left(  S_{y}-V\left(  T_{x}\right)  \right)  \in\Omega(G)$,
otherwise we obtain the following contradiction
\[
\left\vert S_{y}-V(T_{x})\right\vert +\left\vert A\right\vert <\alpha
(G)\leq\alpha(G-T_{x})+\alpha(T_{x})=\left\vert S_{y}-V(T_{x})\right\vert
+\left\vert A\right\vert .
\]

\textit{Case 2}. Assume now that $x\in A$.

Then we have $\left\vert A\right\vert \geq\left\vert S_{yz}\cap V\left(
T_{x}\right)  \right\vert $, because $S_{yz}\cap V\left(  T_{x}\right)  $ is
independent in $T_{x}$. Hence we infer%
\begin{gather*}
\alpha\left(  G\right)  =\left\vert S_{yz}-\left\{  y\right\}  \right\vert
\leq\left\vert \left(  S_{yz}-\left\{  y\right\}  -\left(  S_{yz}\cap V\left(
T_{x}\right)  \right)  \right)  \cup A\right\vert =\\
=\left\vert \left(  S_{yz}-\left\{  y\right\}  -V\left(  T_{x}\right)
\right)  \cup A\right\vert .
\end{gather*}

Since $W=\left(  S_{yz}-\left\{  y\right\}  -V\left(  T_{x}\right)  \right)
\cup A$ is independent and its size is $\alpha\left(  G\right)  $ at least, it
follows that $W$\ is also a maximum independent set, i.e., we have $A\subseteq
W\in\Omega(G)$, as needed.

\textit{Claim 2}. $S\cap V\left(  T_{x}\right)  \in\Omega\left(  T_{x}\right)
$ for every $S\in\Omega\left(  G\right)  $ and each $x\in N_{1}(C)$.

Let $S\in\Omega\left(  G\right)  $, and suppose, to the contrary, that
$A=S\cap V\left(  T_{x}\right)  \notin\Omega\left(  T_{x}\right)  $. By Lemma
\ref{lem5}\emph{, }$x\notin\mathrm{core}(T_{x})$. Thus we can change $A$ for
some $B\in\Omega\left(  T_{x}\right)  $ not containing $x$. The set $\left(
S-A\right)  \cup B$ is clearly independent in $G$, and this leads to the
contradiction $\left\vert \left(  S-A\right)  \cup B\right\vert =\left\vert
S-A\right\vert +\left\vert B\right\vert >\left\vert S\right\vert =\alpha(G)$.

Combining \textit{Claims 1}\emph{ }and\emph{ }\textit{2}, we infer that:%
\begin{align*}
\mathrm{core}\left(  T_{x}\right)   &  =\cap\{A:A\in\Omega(T_{x}%
)\}=\cap\{S\cap V\left(  T_{x}\right)  :S\in\Omega(G)\}\\
&  =(\cap\{S:S\in\Omega(G)\})\cap V\left(  T_{x}\right)  =\mathrm{core}\left(
G\right)  \cap V\left(  T_{x}\right)  ,
\end{align*}
which clearly implies
\[
\mathrm{core}\left(  G\right)  =\cup\left\{  \mathrm{core}\left(
T_{x}\right)  :x\in N(V(C))-V(C)\right\}
\]
as required.
\end{proof}

\begin{remark}
The assertion in Theorem \ref{th22} may fail for:

\emph{(i)} bipartite unicyclic graphs; for example, the graphs $H_{1}$,
$H_{2}$ from Figure \ref{fig1} satisfy
\begin{align*}
\mathrm{core}\left(  H_{1}\right)   &  =\cup\left\{  \mathrm{core}\left(
T_{x}\right)  :x\in N_{1}(C)\right\}  \text{, and }\\
\mathrm{core}\left(  H_{2}\right)   &  \neq\left\{  x,z\right\}  =\cup\left\{
\mathrm{core}\left(  T_{x}\right)  :x\in N_{1}(C)\right\}  ;
\end{align*}
\begin{figure}[h]
\setlength{\unitlength}{1cm}\begin{picture}(5,1.3)\thicklines
\multiput(3,0)(1,0){5}{\circle*{0.29}}
\multiput(4,1)(1,0){4}{\circle*{0.29}}
\put(3,0){\line(1,0){4}}
\put(4,0){\line(0,1){1}}
\put(5,0){\line(0,1){1}}
\put(5,1){\line(1,0){1}}
\put(6,0){\line(0,1){1}}
\put(7,0){\line(0,1){1}}
\put(3,0.3){\makebox(0,0){$a$}}
\put(4.3,1){\makebox(0,0){$b$}}
\put(2.2,0.5){\makebox(0,0){$H_{1}$}}
\multiput(9,0)(1,0){4}{\circle*{0.29}}
\multiput(9,1)(1,0){4}{\circle*{0.29}}
\put(9,0){\line(1,0){2}}
\put(9,1){\line(1,-1){1}}
\put(10,0){\line(0,1){1}}
\put(10,1){\line(1,0){2}}
\put(11,0){\line(0,1){1}}
\put(12,0){\line(0,1){1}}
\put(9,0.3){\makebox(0,0){$x$}}
\put(9.3,1){\makebox(0,0){$y$}}
\put(10.3,0.8){\makebox(0,0){$z$}}
\put(11.3,0){\makebox(0,0){$t$}}
\put(8.2,0.5){\makebox(0,0){$H_{2}$}}
\end{picture}\caption{$H_{1},H_{2}$ are bipartite unicyclic graphs,
\textrm{core}$(H_{1})=\left\{  a,b\right\}  $, \textrm{core}$(H_{2})=\left\{
t,x,y,z\right\}  $.}%
\label{fig1}%
\end{figure}
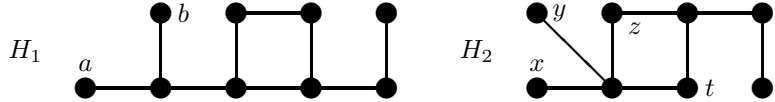

\emph{(ii)} non-bipartite K\"{o}nig-Egerv\'{a}ry unicyclic graphs; for
instance,
\begin{align*}
\mathrm{core}\left(  G_{2}\right)   &  \neq\left\{  t,z\right\}  =\cup\left\{
\mathrm{core}\left(  T_{x}\right)  :x\in N_{1}(C)\right\}  \text{, while }\\
\mathrm{core}\left(  G_{1}\right)   &  =\cup\left\{  \mathrm{core}\left(
T_{x}\right)  :x\in N_{1}(C)\right\}  ,
\end{align*}
where $G_{1}$ and $G_{2}$ are from Figure \ref{fig11222}.\begin{figure}[h]
\setlength{\unitlength}{1cm}\begin{picture}(5,1.2)\thicklines
\multiput(2,0)(1,0){6}{\circle*{0.29}}
\multiput(3,1)(1,0){3}{\circle*{0.29}}
\put(2,0){\line(1,0){5}}
\put(3,0){\line(0,1){1}}
\put(4,0){\line(0,1){1}}
\put(4,1){\line(1,0){1}}
\put(5,1){\line(1,-1){1}}
\put(2,0.3){\makebox(0,0){$a$}}
\put(3.3,1){\makebox(0,0){$b$}}
\put(7,0.3){\makebox(0,0){$c$}}
\put(1.2,0.5){\makebox(0,0){$G_{1}$}}
\multiput(9,0)(1,0){5}{\circle*{0.29}}
\multiput(11,1)(1,0){2}{\circle*{0.29}}
\put(9,0){\line(1,0){4}}
\put(10,0){\line(1,1){1}}
\put(11,1){\line(1,0){1}}
\put(12,0){\line(0,1){1}}
\put(9,0.3){\makebox(0,0){$t$}}
\put(11,0.3){\makebox(0,0){$y$}}
\put(13,0.3){\makebox(0,0){$z$}}
\put(8.2,0.5){\makebox(0,0){$G_{2}$}}
\end{picture}\caption{$G_{1},G_{2}$ are K\"{o}nig-Egerv\'{a}ry graphs,
\textrm{core}$(G_{1})=\left\{  a,b,c\right\}  $, \textrm{core}$(G_{2}%
)=\left\{  t,y,z\right\}  $.}%
\label{fig11222}%
\end{figure}
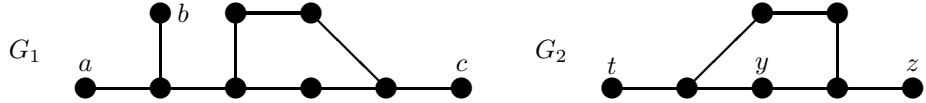
\end{remark}

It is worth mentioning that the problem of whether there are vertices in a
given graph $G$ belonging to $\mathrm{core}\left(  G\right)  $ is
\textbf{NP}-hard \cite{BorGolLev}. In \cite{LevManTAMC} we have presented both
sequential and parallel algorithms finding $\mathrm{core}\left(  G\right)  $
in polynomial time for K\"{o}nig-Egerv\'{a}ry graphs. By Theorem \ref{th22}, a
unicyclic graph is either a K\"{o}nig-Egerv\'{a}ry graph or its $\mathrm{core}%
\left(  G\right)  $ equals a union of cores of a finite number of some special
subtrees. Therefore, we get the following.

\begin{corollary}
If $G$ is a unicyclic graph, then $\mathrm{core}\left(  G\right)  $ is
computable in polynomial time.
\end{corollary}

\section{Conclusions}

The main purpose of this paper is to investigate the structure of
$\mathrm{core}\left(  G\right)  $ for unicyclic graphs. One the one hand, we
have succeeded to represent $\mathrm{core}\left(  G\right)  $ as the union of
cores of some specific subtrees of a non K\"{o}nig-Egerv\'{a}ry unicyclic
graph $G$. On the other hand, it is still not clear if there exists a
characterization of this kind for bipartite unicyclic graphs and/or
non-bipartite K\"{o}nig-Egerv\'{a}ry graphs.

\end{document}